\documentclass[reqno,12pt]{amsart}
\usepackage{amsmath}
\usepackage{amsfonts}
\usepackage{hyperref}

\newtheorem{theorem}{Theorem}
\theoremstyle{plain}

\newtheorem{proposition}{Proposition}
\newtheorem{remark}{Remark}

\numberwithin{equation}{section}
\input{tcilatex}

\begin{document}
\title[Initial data for the characteristic Einstein-Vlasov system]{The
Goursat problem for the Einstein-Vlasov system: (I) The initial data
constraints}
\author[ C. TADMON]{ \textbf{Calvin TADMON}\\
Department of Mathematics and Computer Science, Faculty of
Science, University of Dschang, P.\ O.\ Box 67, Dschang, Cameroon}
\address{}
\email{}
\date{}

\begin{abstract}
We show how to assign, on two intersecting null hypersurfaces, initial data
for the Einstein-Vlasov system in harmonic coordinates. As all the
components of the metric appear in each component of the stress-energy
tensor, the hierarchical method of Rendall can not apply strictly speaking.
To overcome this difficulty, an additional assumption have been imposed to
the metric on the initial hypersurfaces. Consequently, the distribution
function is constrained to satisfy some integral equations on the initial
hypersurfaces.
\end{abstract}

\maketitle

\numberwithin{equation}{section}

\section{Introduction}

This work is devoted to the resolution of the constraints problem
associated to the characteristic Einstein-Vlasov (EV) system on
two intersecting null hypersurfaces. The interests and physical
motivations for studying such problems have been widely mentioned
in \cite{2, 4, 5, 6, 10, 12, 13, 14, 17, 18}. It is well known
that the EV system is not an evolution system as it stands. In
order to obtain a hyperbolic system, one needs to impose some
supplementary conditions called gauge conditions which, due to the
deep structure of the system, must satisfy the following
properties:

$\left( i\right) $ whenever these gauge conditions are fulfilled everywhere
in the space-time, the EV system reduces to a non-linear hyperbolic system
called the evolution system.

$\left( ii\right) $ whenever the associated evolution system is satisfied
everywhere in the space-time and the gauge conditions are satisfied on the
null hypersurfaces that carry the initial data, then these gauge conditions
and the complete EV system are satisfied everywhere.

It therefore follows that when the choice of gauge conditions is made, the
initial value problem for the EV system is naturally decomposed into two
parts called the evolution problem and the constraints problem.

The resolution of the evolution problem is equivalent to the resolution of
the reduced non-linear hyperbolic system obtained from the EV system thanks
to the choice of the gauge conditions. Due to the gauge conditions the data
for the reduced EV system can not be given freely. It is necessary to
construct, from arbitrary choice of some components of the gravitational
potentials (called free data) on the initial null hypersurfaces, all the
initial data such that the solution of the reduced EV system with those
initial data satisfies the gauge conditions on the initial null
hypersurfaces. The construction of such data is referred to as the
resolution of the constraints problem. Through all the work we will use
harmonic gauge for the gravitational field.

We now proceed to survey some relatively recent works known about
characteristic initial value problems with initial data prescribed on two
intersecting null hypersurfaces often referred to as the Goursat problems.
In 1990, A. D. Rendall \cite{17} published a $C^{\infty }$ existence and
uniqueness result for quasilinear hyperbolic systems of second order\ with $%
C^{\infty }$ data prescribed on two intersecting null hypersurfaces. Using
the harmonic gauge, the author applied the $C^{\infty }$\ result obtained in
\cite{17}\ to solve the characteristic initial value problem for the
Einstein equations in vacuum and with relativistic perfect fluid source. For
sake of more physical applications, it is known that, for Partial
Differential Equations (PDE), solutions of finite differentiability order
are more important than those of infinite differentiability order. In \cite%
{17} section 7 the author mentioned briefly how results of finite
differentiability order can be obtained for data of finite differentiability
order although proofs were not given. In 1990, H.\ M\"{u}ller zum Hagen \cite%
{14} used Sobolev type inequalities to derive energy inequalities that
enable him solve, in weighted Sobolev space (results of finite
differentiability order), the characteristic initial value problem for
linear hyperbolic systems of second order. He also predicted an existence
and uniqueness result for the quasilinear case. Apart from the fundamental
papers \cite{14} and \cite{17}, some other works on characteristic initial
value problems with initial data prescribed on two intersecting null
hypersurfaces can be found in \cite{2, 3, 8, 9, 10, 12, 18}. As pointed out
by H.\ Andreasson \cite{1}, A. D.\ Rendall \cite{17} and M.\ Fjallborg \cite%
{11}, unlike some known models, the Einstein-Vlasov model has a very nice
feature in General Relativity and Kinetic Theory since the stress-energy
tensor fulfills, without any supplementary assumption, all the physical
necessary energy conditions i.e. the weak energy condition, the dominant
energy condition and the strong energy condition as well as the non-negative
sum pressures condition. This situation, coupled with the importance of
characteristic initial value problems mentioned at the beginning, motivates
us to study the constraints problem associated to the characteristic EV
system. When attempting to solve the constraints problem for the
characteristic EV system by the hierarchical method of Rendall (see \cite%
{17, 18}), a crucial obstacle occurs due to the complicated form of each\
component\ of\ the stress-energy tensor where all the components of the
metric to be constructed appear. The novelty of our work resides in the fact
that we have worked out this difficulty through a supplementary judicious
assumption imposed to the gravitational potentials on the initial
hypersurfaces. As a consequence of the additional assumption on the
gravitational potentials, the distribution function can not be given as free
data, it must satisfy some integral equations. Another advantage of this
paper is that, unlike the work of Rendall \cite{17, 18}, many delicate
calculations and expressions are given in details in such a way that we can
foresee promising resolution of the global characteristic EV system using
for example tools that are similar to those of G.\ Caciotta and\ F.\ Nicolo
\cite{3,4}. To reduce the length of the paper, the evolution problem for the
characteristic EV system is out of the scope of the present work and will be
solved in a forthcoming paper. The paper is organized as follows. In section
2, we give some preliminaries about the EV system. The complete form as well
as the reduced form (in harmonic coordinates) of the EV system are written.
A new form of EV system is derived with appropriate unknowns an variables.
This new form of EV system is suitable for the resolution of the constraints
problem. The concern of section 3 is the resolution of the constraints
problem for the characteristic EV system i.e. the construction of the
initial data for the reduced EV system such that the harmonic gauge
conditions are satisfied on the initial null hypersurfaces. For sake of
simplicity and clarity, only the case of $C^{\infty }$ data will be
discussed. Data of finite differentiability order may be constructed in
Sobolev type spaces using energy inequalities and other classical tools as
described in \cite{9, 10, 14, 17} and references$\ $therein. An appendix D
is provided at the end of the work and is devoted to the treatment of the
constraints integral equations which must be satisfied by the distribution
function. It would be of interest to investigate whether the additional
assumption on free data as well as the constraints integral equations have a
particular physical meaning.

\section{The Einstein-Vlasov (EV) system}

\subsection{The complete form of the EV system}

The geometric framework is a four dimensional differentiable manifold $%
\mathcal{M}$, endowed with a hyperbolic metric $\widehat{g}$ of signature $%
-+++$. The manifold $\left( \mathcal{M},\widehat{g}\right) $ is called a
space-time. $\mathcal{M}$ is assumed to be orientable and of class $%
C^{\infty }$. Throughout the remainder of the work, commas will be used to
denote partial derivatives e.g. $\widehat{g}_{ij,k}=\frac{\partial \widehat{g%
}_{ij}}{\partial y^{k}}$. Roman indices $i,j,...$ run from $1$ to $4$ while
Greek ones $\alpha ,\beta ,...$\ run from $3$ to $4.$ Einstein convention on
repeated indices is used i.e $A_{i}B^{i}=\underset{i}{\sum }A_{i}B^{i}$. The
Einstein-Vlasov system is written as follows (see \cite{1, 7, 11, 19})

\begin{equation}
\begin{array}{l}
\widehat{S}_{ij}\equiv \widehat{R}_{ij}-\frac{1}{2}\widehat{R}\widehat{g}%
_{ij}=\widehat{T}_{ij}, \\
q^{i}\frac{\partial f}{\partial y^{i}}-\widehat{\Gamma }_{jk}^{i}q^{j}q^{k}%
\frac{\partial f}{\partial q^{i}}=0,%
\end{array}
\tag{2.1}  \label{5.1}
\end{equation}%
where $\widehat{g}_{ij}$ are the covariant components of the metric $%
\widehat{g}$. They constitute the unknowns for the Einstein equations. $%
\widehat{R}_{ij}$ are the covariant components of the Ricci tensor and $%
\widehat{R}$ is the scalar curvature of the metric $\widehat{g}$. In the
local coordinates $\left( y^{i}\right) $ they are given as follows

\begin{equation}
\widehat{R}_{ij}=\widehat{R}_{ikj}^{k}=\widehat{\Gamma }_{ij,k}^{k}-\widehat{%
\Gamma }_{ik,j}^{k}+\widehat{\Gamma }_{kl}^{k}\widehat{\Gamma }_{ij}^{l}-%
\widehat{\Gamma }_{jl}^{k}\widehat{\Gamma }_{ik}^{l},\quad \widehat{R}=%
\widehat{g}^{ij}\widehat{R}_{ij},  \tag{2.2}  \label{5.2}
\end{equation}%
where $\widehat{\Gamma }_{ij}^{k}$ are the Christoffel symbols of the metric
$\widehat{g}$ i.e.,

\begin{equation}
\widehat{\Gamma }_{jk}^{i}=\frac{1}{2}\widehat{g}^{il}\left( \widehat{g}%
_{lk,j}+\widehat{g}_{lj,k}-\widehat{g}_{jk,l}\right) ,  \tag{2.3}
\label{5.3}
\end{equation}%
$\widehat{g}^{il}$\ are the contravariant components of $\widehat{g}$\ i.e.,%
\textit{\ }%
\begin{equation}
\widehat{g}^{il}\widehat{g}_{lk}=\delta _{k}^{i}=\left\{
\begin{array}{l}
1\text{ for }i=k, \\
0\text{ for }i\neq k.%
\end{array}%
\right.  \tag{2.4}  \label{5.4}
\end{equation}%
$f$\ is the distribution function (or the particle number density function)
which constitutes the unknown for the Vlasov equation. $f$\ is a
non-negative real valued function defined on $F\left( \mathcal{M}\right) $,
where
\begin{equation}
F\left( \mathcal{M}\right) =\underset{y\in \mathcal{M}}{\cup }\left\{
q=\left( q^{i}\right) \in T_{y}:\widehat{g}_{ij}\left( y\right)
q^{i}q^{j}=-m^{2},\text{ }0<q^{1}\right\} ,  \tag{2.5}  \label{5.5}
\end{equation}%
with $T_{y}\equiv T_{y}\mathcal{M}.$ The Vlasov equation symbolizes the
conservation of the number of particles along the trajectories across the
hypersurfaces of $F\left( \mathcal{M}\right) $ in the case where there is no
collision between particles. $\widehat{T}_{ij}$ are the covariant components
of the stress-energy (or energy-momentum) tensor which is the source of the
gravitational field created by the particles. In contravariant components
the stress-energy is defined by the following relation (see \cite{7})

\begin{equation}
\widehat{T}^{ij}\left( y\right) =-\int_{F_{y}}f\left( y,q\right) q^{i}q^{j}%
\frac{\left\vert \widehat{g}\right\vert ^{\frac{1}{2}}}{q_{1}}d^{3}q\text{,}
\tag{2.6}  \label{5.6}
\end{equation}%
where $F_{y}=\left\{ q=\left( q^{i}\right) \in T_{y}:\widehat{g}_{ij}\left(
y\right) q^{i}q^{j}=-m^{2},\text{ }0<q^{1}\right\} $, $d^{3}q=dq^{2}\wedge
dq^{3}\wedge dq^{4}$, $\left\vert \widehat{g}\right\vert $ is the modulus of
the determinant of $\left( \widehat{g}_{ij}\right) $.

\subsection{The reduced EV system}

The Einstein equations as they stand are not hyperbolic but in harmonic
coordinates they read (see \cite{7})

\begin{equation}
\widehat{R}_{ij}^{h}=\widehat{T}_{ij},  \tag{2.7}  \label{5.7}
\end{equation}%
\ where
\begin{equation}
\widehat{R}_{ij}^{h}\equiv \widehat{R}_{ij}-\frac{1}{2}\left( \widehat{g}%
_{ik}\widehat{\Gamma }_{,j}^{k}+\widehat{g}_{jk}\widehat{\Gamma }%
_{,i}^{k}\right) =-\frac{1}{2}\widehat{g}^{km}\widehat{g}_{ij,mk}+Q_{ij}.
\tag{2.8}  \label{5.8}
\end{equation}%
Here $Q_{ij}$\ is a rational function depending on the metric components and
their first order derivatives (see \cite{10, 20}),
\begin{equation}
\widehat{\Gamma }^{k}=\widehat{g}^{ij}\widehat{\Gamma }_{ij}^{k}.  \tag{2.9}
\label{5.9}
\end{equation}%
So the reduced EV system reads%
\begin{equation}
\begin{array}{l}
-\frac{1}{2}\widehat{g}^{km}\widehat{g}_{ij,mk}+Q_{ij}=\widehat{T}_{ij}, \\
q^{i}\frac{\partial f}{\partial y^{i}}+Q^{i}\frac{\partial f}{\partial q^{i}}%
=0,%
\end{array}
\tag{2.10}  \label{5.10}
\end{equation}%
where
\begin{equation}
Q^{i}=-\widehat{\Gamma }_{jk}^{i}q^{j}q^{k}.  \tag{2.11}  \label{5.11}
\end{equation}

\subsection{Appropriate unknowns and variables}

As a relativistic speed is bounded, we think that it is convenient to choose
on the mass shell, coordinates with bounded domain (see also \cite{7}). Let $%
y\in \mathbb{U}$, $\mathbb{U}$ is the domain of a local chart in $\mathcal{M}
$. Set $w^{A}=\frac{q^{A}}{q^{1}}$, $A=2,3,4$ and denote by $M_{y}$ the
image in $\mathbb{R}^{3}$ of $F_{y}$ by the mapping $\left( q^{i}\right)
\mapsto \left( w^{A}\right) $. Assume the following hyperbolicity conditions
on $\left( \widehat{g}_{ij}\right) .$

\textbf{Assumption }$\widehat{h}$: The metric $\left( \widehat{g}%
_{ij}\right) $ is uniformly hyperbolic and the hypersurfaces $y^{1}=$ $Cte$
are uniformly spatial i.e.

\begin{equation}
\begin{array}{l}
\exists a,b\in \left( 0,\infty \right) :a^{2}\left\vert \xi \right\vert
^{2}\leq \widehat{g}_{AB}\xi ^{A}\xi ^{B}\leq b^{2}\left\vert \xi
\right\vert ^{2}\text{ where }\left\vert \xi \right\vert ^{2}=\overset{4}{%
\underset{A=2}{\sum }}\left( \xi ^{A}\right) ^{2}, \\
-\widehat{g}_{11}\geq a^{2}\text{ and }-\widehat{g}^{11}\geq a^{2}.%
\end{array}
\tag{2.12}  \label{5.12}
\end{equation}

\begin{proposition}
$\left( i\right) $ Under assumption $\left( \ref{5.12}\right) $, the Vlasov
equation reads%
\begin{equation}
q^{i}\frac{\partial f}{\partial y^{i}}+Q^{A}\frac{\partial f}{\partial q^{A}}%
=0.  \tag{2.13}  \label{5.13}
\end{equation}%
\ $\left( ii\right) $ Under assumption $\left( \ref{5.12}\right) $, $M_{y}$\
is a bounded domain in $\mathbb{R}^{3}$\ such that $M_{y}\subset M$, where $%
M $\ is a fixed compact domain in $\mathbb{R}^{3}$. The stress-energy tensor
is given as follows%
\begin{equation}
\widehat{T}^{ij}\left( y\right) =\frac{1}{m^{2}}\int_{M_{y}}f\left(
y,w\right) q^{i}q^{j}\left( q^{1}\right) ^{4}\left\vert \widehat{g}%
\right\vert ^{\frac{1}{2}}d^{3}w,  \tag{2.14}  \label{5.14}
\end{equation}%
where $d^{3}w=dw^{2}\wedge dw^{3}\wedge dw^{4}$, $f\left( y,w\right) $
\textit{is the expression of} $f\left( y,q\right) $ \textit{in the local
coordinates} $\left( y,w\right) $.
\end{proposition}

\begin{proof}
See\textbf{\ }\cite{7}.
\end{proof}

\begin{remark}
\textit{In the expression }$\left( \ref{5.14}\right) $\textit{\ of the\
stress-energy tensor, we would like to write }$f\left( y,w\right)
q^{i}q^{j}\left( q^{1}\right) ^{4}$\textit{\ as }$\varphi \left( y,w\right)
w^{ij},$\textit{\ where }$w^{ij}$\textit{\ does not depend on} $\left(
\widehat{g}_{ij}\right) $.\textit{\ To do so, we proceed to the following
change of the unknown distribution function by setting }$f\left( y,w\right)
=\varphi \left( y,w\right) \left( q^{1}\right) ^{-6}$.\textit{\ So we must
have }$w^{ij}$\textit{\ }$=\frac{q^{i}q^{j}}{\left( q^{1}\right) ^{2}}.$
\end{remark}

\begin{proposition}
\textit{Under the change }$f\left( y,w\right) =\varphi \left( y,w\right)
\left( q^{1}\right) ^{-6}$,\textit{\ the stress-energy tensor is given as
follows}

\begin{equation}
\widehat{T}^{ij}\left( y\right) =\frac{1}{m^{2}}\int_{M_{y}}\varphi \left(
y,w\right) w^{ij}\left\vert \widehat{g}\right\vert ^{\frac{1}{2}}d^{3}w\text{%
,}  \tag{2.15}  \label{5.16}
\end{equation}%
\textit{where}%
\begin{equation*}
w^{ij}=\frac{q^{i}q^{j}}{\left( q^{1}\right) ^{2}}.
\end{equation*}%
\textit{The Vlasov equation becomes}

\begin{equation}
q^{i}\frac{\partial \varphi }{\partial y^{i}}+\frac{1}{q^{1}}\left(
Q^{A}-w^{A}Q^{1}\right) \frac{\partial \varphi }{\partial w^{A}}-\frac{6}{%
q^{1}}Q^{1}\varphi =0.  \tag{2.16}  \label{5.17}
\end{equation}
\end{proposition}

\begin{proof}
See \cite{7}.
\end{proof}

\begin{remark}
\textit{The expression }$\left( \ref{5.16}\right) $\textit{\ of the
stress-energy tensor is not appropriate since the domain }$M_{y}$\textit{\
depends on }$y$\textit{\ and makes it difficult to differentiate }$\widehat{T%
}^{ij}$\textit{\ even in the distributional sense. It appears therefore
judicious to transform this domain in order to make it independent of }$y$.
\end{remark}

Assume the following decomposition of $\widehat{g}$.

\textbf{Assumption }$\widehat{h}^{\prime }$: The spatial part of $\left(
\widehat{g}_{AB}\right) $ is decomposed as follows

\begin{equation}
\widehat{g}_{AB}Y^{A}Y^{B}=\overset{4}{\underset{B=2}{\sum }}\left( \lambda
_{A}^{B}Y^{A}\right) ^{2},  \tag{2.17}  \label{5.18}
\end{equation}%
where $\lambda _{A}^{B}$ are functions that depends smoothly ($C^{\infty }$
for instance) on the components $\widehat{g}_{AB}$\ of the metric. Set

\begin{equation}
v^{C}=\left( -\widehat{g}^{11}\right) ^{\frac{1}{2}}\lambda _{A}^{C}\left[
w^{A}+\widetilde{g}^{1A}\right] ,  \tag{2.18}  \label{5.19}
\end{equation}%
with%
\begin{equation}
\widetilde{g}^{1A}=-\frac{1}{\widehat{g}^{11}}\widehat{g}^{1A}.  \tag{2.19}
\label{5.20}
\end{equation}%
$\left( \ref{5.19}\right) $ is equivalent to the following relation

\begin{equation}
w^{A}=\left( \lambda _{B}^{A}\right) ^{-1}v^{B}\left( -\widehat{g}%
^{11}\right) ^{-\frac{1}{2}}-\widetilde{g}^{1A},  \tag{2.20}  \label{5.21}
\end{equation}%
where $\left( \lambda _{B}^{A}\right) ^{-1}$ are the components of the
inverse of the matrix $\left( \lambda _{J}^{I}\right) $.

\begin{proposition}
\textit{The image of }$M_{y}$\textit{\ by the mapping }$\left( w^{A}\right)
\mapsto \left( v^{A}\right) $\textit{\ is the unit open ball }$B$\textit{\
in }$\mathbb{R}^{3}$\textit{. In the parameters }$\left( v^{A}\right) $%
\textit{, the energy-momentum tensor reads}%
\begin{equation}
\widehat{T}^{ij}\left( y\right) =\frac{1}{m^{2}}\int_{B}\varphi \left(
y,v\right) v^{ij}\left\vert \widehat{g}\right\vert ^{\frac{1}{2}}\left( -%
\widehat{g}^{11}\right) ^{-\frac{3}{2}}\left\vert \widetilde{g}\right\vert
^{-\frac{1}{2}}d^{3}v\text{,}  \tag{2.21}  \label{5.22}
\end{equation}%
where $d^{3}v=dv^{2}\wedge dv^{3}\wedge dv^{4},$ $v^{ij}=\frac{q^{i}q^{j}}{%
\left( q^{1}\right) ^{2}}$, $\left\vert \widetilde{g}\right\vert $ \textit{%
is the modulus of the determinant of} $\left( \widehat{g}_{AB}\right) $, $%
\varphi \left( y,v\right) $ \textit{is the expression of} $\varphi \left(
y,w\right) $ \textit{in the local coordinates} $\left( y,v\right) $. \textit{%
The Vlasov equation reads as follows}%
\begin{equation}
\begin{array}{l}
\frac{\partial \varphi }{\partial y^{1}}+\left[ \left( \lambda
_{B}^{A}\right) ^{-1}v^{B}\left( -\widehat{g}^{11}\right) ^{-\frac{1}{2}}-%
\widetilde{g}^{1A}\right] \frac{\partial \varphi }{\partial y^{A}} \\
+\left( q^{1}\right) ^{-2}\left( -\widehat{g}^{11}\right) ^{\frac{1}{2}%
}\lambda _{B}^{A}\left( Q^{B}-w^{B}Q^{1}\right) \frac{\partial \varphi }{%
\partial v^{A}}-6\left( q^{1}\right) ^{-2}Q^{1}\varphi =0.%
\end{array}
\tag{2.22}  \label{5.24}
\end{equation}
\end{proposition}

\begin{proof}
See \cite{7}.
\end{proof}

\begin{remark}
\textit{Y. Choquet-Bruhat }\cite{7}\textit{\ used assumption }$\left( \ref%
{5.18}\right) $\textit{\ and a variant of assumption }$\left( \ref{5.12}%
\right) $\textit{\ to treat the ordinary Cauchy problem for the EV system.
But in the characteristic case, those assumptions are not appropriate and
they need to be recast. We proceed to the desired adaptation through a
judicious change of local events variables }$\left( y^{i}\right) $.
\end{remark}

\begin{proposition}
\textit{Let }$\left( y^{i}\right) $\textit{\ be a local coordinates system
on }$\mathcal{M}$\textit{\ in which the components }$\left( \widehat{g}%
_{ij}\right) $ \textit{of the metric satisfy assumption }$\left( \ref{5.12}%
\right) $\textit{. Set}%
\begin{equation}
\begin{array}{l}
x^{1}=\frac{1}{2}\left( y^{1}+y^{2}\right) ,\quad x^{2}=\frac{1}{2}\left(
y^{1}-y^{2}\right) , \\
x^{\alpha }=y^{\alpha }\text{,}\quad \alpha =3,4.%
\end{array}
\tag{2.23}  \label{5.25}
\end{equation}%
\textit{In the coordinates system }$\left( x,p\right) $\textit{, the EV
system reads}%
\begin{equation}
\begin{array}{l}
S_{ij}\equiv R_{ij}-\frac{1}{2}Rg_{ij}=T_{ij}\text{,} \\
p^{i}\frac{\partial f}{\partial x^{i}}+P^{i}\frac{\partial f}{\partial p^{i}}%
=0\text{,}%
\end{array}
\tag{2.24}  \label{5.26}
\end{equation}%
\textit{where}%
\begin{equation}
\begin{array}{l}
g_{ij}\left( x\right) =\frac{\partial y^{k}}{\partial x^{i}}\frac{\partial
y^{l}}{\partial x^{j}}\widehat{g}_{kl}\left( y\right) ,\text{ }R_{ij}\left(
x\right) =\frac{\partial y^{k}}{\partial x^{i}}\frac{\partial y^{l}}{%
\partial x^{j}}\widehat{R}_{kl}\left( y\right) , \\
R\left( x\right) =g^{ij}\left( x\right) R_{ij}\left( x\right) ,\text{ }%
T_{ij}\left( x\right) =\frac{\partial y^{k}}{\partial x^{i}}\frac{\partial
y^{l}}{\partial x^{j}}\widehat{T}_{kl}\left( y\right) , \\
p^{i}=\frac{\partial x^{i}}{\partial y^{k}}q^{k},\text{ }P^{i}=\frac{%
\partial x^{i}}{\partial y^{k}}Q^{k}=-\Gamma _{jk}^{i}p^{j}p^{k}\text{, } \\
\Gamma _{jk}^{i}=\frac{1}{2}g^{ic}\left( g_{ck,j}+g_{cj,k}-g_{jk,c}\right) .%
\end{array}
\tag{2.25}  \label{5.27}
\end{equation}
\end{proposition}

\begin{proof}
A\ direct\ calculation\ leads\ to\ the\ desired\ equations.\textbf{\ }
\end{proof}

\begin{remark}
\textit{The change of local coordinates }$\left( \ref{5.25}\right) $\textit{%
\ preserves the harmonicity. In other words, we have the following
equivalence}%
\begin{equation}
\left( \forall i=1,2,3,4,\text{ }\widehat{\Gamma }^{i}\equiv \widehat{g}^{kl}%
\widehat{\Gamma }_{kl}^{i}=0\right) \Leftrightarrow \left( \forall i=1,2,3,4,%
\text{ }\Gamma ^{i}\equiv g^{kl}\Gamma _{kl}^{i}=0\right) .  \tag{2.26}
\label{5.28}
\end{equation}%
\textit{In fact, the following relations hold}%
\begin{equation}
\begin{array}{l}
\widehat{\Gamma }^{1}=\Gamma ^{1}+\Gamma ^{2},\quad \widehat{\Gamma }%
^{2}=\Gamma ^{1}-\Gamma ^{2}, \\
\widehat{\Gamma }^{\alpha }=\Gamma ^{\alpha }\text{,}\quad \alpha =3,4.%
\end{array}
\tag{2.27}  \label{5.29}
\end{equation}
\end{remark}

\begin{proposition}
\textit{The stress-energy tensor }$\left( \ref{5.22}\right) $\textit{\ is
given in the local coordinates }$\left( x^{i}\right) $ \textit{and the local
parameters }$\left( v^{A}\right) $\textit{\ as follows}%
\begin{equation}
T^{ij}\left( x\right) =\frac{1}{2m^{2}}\int_{B}\varphi \left( x,v\right)
v^{ij}\left\vert g\right\vert ^{\frac{1}{2}}\left( -\widehat{g}^{11}\right)
^{-\frac{3}{2}}\left\vert \widetilde{g}\right\vert ^{-\frac{1}{2}}d^{3}v%
\text{, }  \tag{2.28}  \label{5.30}
\end{equation}%
\textit{where }$v^{ij}=\frac{p^{i}p^{j}}{\left( q^{1}\right) ^{2}}$, $p^{i}=%
\frac{\partial x^{i}}{\partial y^{k}}q^{k}$, $\left\vert g\right\vert $
\textit{is the modulus of the determinant of} $\left( g_{ij}\right) $, $%
\left\vert \widetilde{g}\right\vert $ \textit{is the modulus of the
determinant of} $\left( \widehat{g}_{AB}\right) $, $\varphi \left(
x,v\right) $ \textit{is the expression of} $\varphi \left( y,v\right) $
\textit{in the local coordinates} $\left( x,v\right) $. \textit{The Vlasov
equation }$\left( \ref{5.24}\right) $\textit{\ becomes}%
\begin{equation}
H^{i}\frac{\partial \varphi }{\partial x^{i}}+L^{C}\frac{\partial \varphi }{%
\partial v^{C}}+F\varphi =0,  \tag{2.29}  \label{5.32}
\end{equation}%
\textit{with}%
\begin{equation}
\begin{array}{l}
w^{A}=\left( \lambda _{B}^{A}\right) ^{-1}v^{B}\left( -\widehat{g}%
^{11}\right) ^{-\frac{1}{2}}-\widetilde{g}^{1A},\quad \widehat{g}^{11}\left(
y\right) =g^{11}\left( x\right) +2g^{12}\left( x\right) +g^{22}\left(
x\right) , \\
H^{1}\left( x,v\right) =\frac{1}{2}\left( 1+w^{2}\right) ,\quad H^{2}\left(
x,v\right) =\frac{1}{2}\left( 1-w^{2}\right) ,\quad H^{\alpha }\left(
x,v\right) =w^{\alpha },\quad \alpha =3,4, \\
L^{C}\left( x,v\right) =\frac{1}{4}\left( -\widehat{g}^{11}\right) ^{\frac{1%
}{2}}\lambda _{2}^{C}l^{2}\left( x,v\right) -\frac{1}{4}\left( -\widehat{g}%
^{11}\right) ^{\frac{1}{2}}\lambda _{\alpha }^{C}l^{\alpha }\left(
x,v\right) ,\quad F\left( x,v\right) =\frac{3}{2}w^{2}\widetilde{F}.%
\end{array}
\tag{2.30}  \label{5.33}
\end{equation}%
Here%
\begin{eqnarray*}
\widetilde{F} &=&\left( \Gamma _{11}^{2}+\Gamma _{11}^{1}\right) \left(
1+w^{2}\right) ^{2}+\left( \Gamma _{22}^{2}+\Gamma _{22}^{1}\right) \left(
1-w^{2}\right) ^{2}+2\left( \Gamma _{12}^{2}+\Gamma _{12}^{1}\right) \left(
1-\left( w^{2}\right) ^{2}\right) \\
&&+4\left( \Gamma _{1\lambda }^{2}+\Gamma _{1\lambda }^{1}\right) w^{\lambda
}\left( 1+w^{2}\right) +4\left( \Gamma _{2\lambda }^{2}+\Gamma _{2\lambda
}^{1}\right) w^{\lambda }\left( 1-w^{2}\right) +4\left( \Gamma _{\lambda \mu
}^{2}+\Gamma _{\lambda \mu }^{1}\right) w^{\lambda }w^{\mu },
\end{eqnarray*}%
and%
\begin{equation*}
l^{2}\left( x,v\right) =l_{1}^{2}+w^{2}l_{2}^{2},\quad l^{\alpha }\left(
x,v\right) =l_{1}^{\alpha }-w^{\alpha }l_{2}^{2},\quad \alpha =3,4,
\end{equation*}%
with%
\begin{eqnarray*}
l_{1}^{2} &=&\left( \Gamma _{11}^{2}-\Gamma _{11}^{1}\right) \left(
1+w^{2}\right) ^{2}+\left( \Gamma _{22}^{2}-\Gamma _{22}^{1}\right) \left(
1-w^{2}\right) ^{2} \\
&&+2\left( \Gamma _{12}^{2}-\Gamma _{12}^{1}\right) \left( 1-\left(
w^{2}\right) ^{2}\right) +4\left( \Gamma _{1\lambda }^{2}-\Gamma _{1\lambda
}^{1}\right) w^{\lambda }\left( 1+w^{2}\right) \\
&&+4\left( \Gamma _{2\lambda }^{2}-\Gamma _{2\lambda }^{1}\right) w^{\lambda
}\left( 1-w^{2}\right) +4\left( \Gamma _{\lambda \mu }^{2}-\Gamma _{\lambda
\mu }^{1}\right) w^{\lambda }w^{\mu },
\end{eqnarray*}%
\begin{eqnarray*}
l_{2}^{2} &=&\left( \Gamma _{11}^{2}+\Gamma _{11}^{1}\right) \left(
1+w^{2}\right) ^{2}+\left( \Gamma _{22}^{2}+\Gamma _{22}^{1}\right) \left(
1-w^{2}\right) ^{2} \\
&&+2\left( \Gamma _{12}^{2}+\Gamma _{12}^{1}\right) \left( 1-\left(
w^{2}\right) ^{2}\right) +4\left( \Gamma _{1\lambda }^{2}+\Gamma _{1\lambda
}^{1}\right) w^{\lambda }\left( 1+w^{2}\right) \\
&&+4\left( \Gamma _{2\lambda }^{2}+\Gamma _{2\lambda }^{1}\right) w^{\lambda
}\left( 1-w^{2}\right) +4\left( \Gamma _{\lambda \mu }^{2}+\Gamma _{\lambda
\mu }^{1}\right) w^{\lambda }w^{\mu },
\end{eqnarray*}%
\begin{eqnarray*}
l_{1}^{\alpha } &=&\Gamma _{11}^{\alpha }\left( 1+w^{2}\right) ^{2}+\Gamma
_{22}^{\alpha }\left( 1-w^{2}\right) ^{2} \\
&&+2\Gamma _{12}^{\alpha }\left( 1-\left( w^{2}\right) ^{2}\right) +4\Gamma
_{1\lambda }^{\alpha }w^{\lambda }\left( 1+w^{2}\right) \\
&&+4\Gamma _{2\lambda }^{\alpha }w^{\lambda }\left( 1-w^{2}\right) +4\Gamma
_{\lambda \mu }^{\alpha }w^{\lambda }w^{\mu }.
\end{eqnarray*}
\end{proposition}

\begin{proof}
It is straightforward though lengthy.
\end{proof}

\begin{remark}
The EV system $\left( \ref{5.26}\right) $ i\textit{n the local\ coordinates }%
$\left( x,v\right) $ reads as follows%
\begin{equation}
R_{ij}-\frac{1}{2}Rg_{ij}=T_{ij},\quad H^{i}\frac{\partial \varphi }{%
\partial x^{i}}+L^{C}\frac{\partial \varphi }{\partial v^{C}}+F\varphi =0.
\tag{2.31}  \label{5.34}
\end{equation}%
\textit{The reduced EV system }$\left( \ref{5.10}\right) $\textit{\ }i%
\textit{n the local\ coordinates }$\left( x,v\right) $ reads as follows%
\begin{equation}
\begin{array}{c}
\widetilde{R}_{ij}=T_{ij},\quad H^{i}\frac{\partial \varphi }{\partial x^{i}}%
+L^{C}\frac{\partial \varphi }{\partial v^{C}}+F\varphi =0,%
\end{array}
\tag{2.32}  \label{5.35}
\end{equation}%
\textit{where}
\begin{equation}
\widetilde{R}_{ij}\equiv R_{ij}-\frac{1}{2}\left( g_{ki}\Gamma
_{,j}^{k}+g_{kj}\Gamma _{,i}^{k}\right) =-\frac{1}{2}g^{km}g_{ij,mk}+Q_{ij}.
\tag{2.33}  \label{5.35a}
\end{equation}
\end{remark}

\section{The constraints problem for the EV system}

The task here is the construction of initial data for the reduced
EV such that the constraints $\Gamma ^{k}=0$\ are satisfied on
$G^{1}\cup G^{2}$, where $G^{1}$ and $G^{2}$ are the hypersurfaces
in $\mathbb{R}^{4}$ defined by $x^{1}=0$ and $x^{2}=0\
$respectively. We will need $G_{T}^{\omega }=\left\{ x\in
G^{\omega }:0\leq x^{1}+x^{2}\leq T\right\} ,$ $T>0,$ $\omega
=1,2$. Here the problem is much more difficult than in \cite{10,
17}. This difficulty has something to do with the appearance of
all the components of the metric in any component of the
stress-energy tensor. To overcome this toughness, we add a
supplementary assumption on the metric along the initial
hypersurfaces. All the same, we try to mimic, as far as possible,
the hierarchical method of Rendall \cite{17, 18}. We will see that
the supplementary assumption has as consequence to force the
distribution function to satisfy specific integral equations on
the initial hypersurfaces. The resolution of the integral
equations derived can be achieved under suitable conditions (see
appendix D). Nevertheless it is still to be investigated whether
our additional assumption has a particular physical meaning. The
construction will be made in a standard harmonic coordinates
system. The existence of such standard harmonic coordinates system
has been established by A.\ D.\ Rendall \cite{17}.\

Let us now adapt the method of Rendall to construct $C^{\infty }$
initial data for the EV system. The assumptions under\ which the\
work\ is\ achieved\ are\ such\ that\ only\ the\ data\ $g_{\alpha
\beta }$\ and $g_{12}$\ have\ to\ be\ constructed\ on $G^{1}\cup
G^{2}$, the\ relations\ $\Gamma ^{k}=0$ have$\ $to$\ $be$\
$arranged$\ $on $G^{1}\cup G^{2}$, the\ relations\
$g_{22,1}=2g_{12,2}$ and $g_{11,2}=2g_{12,1}$ have\ to\ be
established\ on\ $G^{1}$ and $G^{2}$ respectively. The
construction
of the data is done fully on $G^{1}$ and it will be clear that data on $%
G^{2} $ are constructed in quite a similar way. The first level of the
hierarchy is now described.

\subsection{Construction of\ $g_{\protect\alpha \protect\beta }$\ and $%
g_{12} $\ on $G_{T}^{1}$, relations $\Gamma ^{1}=0$\ and $g_{22,1}=2g_{12,2}$%
}

Let $T\in \left( 0,\infty \right) $, $\left( h_{\alpha \beta }\right) =\left[
\begin{array}{ll}
h_{33} & h_{34} \\
h_{34} & h_{44}%
\end{array}%
\right] $ a matrix with determinant $1$ at each point. Set
$g_{\alpha \beta }=\Omega h_{\alpha \beta }$, where $\Omega >0$ is
an unknown function called the conformity factor. Assume as in
\cite{10, 17} that

\begin{equation}
g_{22}=g_{23}=g_{24}=0\text{ on }G_{T}^{1}.  \tag{3.1}  \label{5.36a}
\end{equation}%
The additional assumption on\ free\ data is the following%
\begin{equation}
g_{11}=g_{13}=g_{14}=0\text{ on }G_{T}^{1}.  \tag{3.2}  \label{5.36b}
\end{equation}%
On $G_{T}^{1}$ it holds

\begin{equation}
\begin{array}{l}
g_{12}g^{12}=1,\quad g^{11}=g^{1\alpha }=0, \\
g^{22}=g^{2\alpha }=0,\quad g_{\lambda \beta }g^{\alpha \beta }=\delta
_{\lambda }^{\alpha }.%
\end{array}
\tag{3.3}  \label{5.37}
\end{equation}

\subsubsection{Expression of $R_{22}$\ and $T_{22}$}

\begin{proposition}
On $G_{T}^{1}$,\ it holds that
\begin{equation}
\begin{array}{l}
R_{22}=\frac{1}{4}g^{12}g^{\alpha \beta }g_{\alpha \beta ,2}\left(
2g_{12,2}-g_{22,1}\right) +\frac{1}{4}g_{,2}^{\beta \lambda }g_{\lambda
\beta ,2}-\frac{1}{2}\left( g^{\alpha \beta }g_{\alpha \beta ,2}\right)
_{,2}, \\
T_{22}=\left( g_{12}\right) ^{4}K_{22},%
\end{array}
\tag{3.4}  \label{5.38}
\end{equation}%
where
\begin{equation}
K_{22}\left( x\right) =\frac{1}{16m^{2}}\int_{B}\varphi \left( x,v\right)
\left( 1+v^{2}\right) ^{2}d^{3}v.  \tag{3.5}
\end{equation}
\end{proposition}

\begin{proof}
See appendix A.
\end{proof}

Assume $g_{2i}=0$\ for $i\neq 1$\ on $G_{T}^{1}$, $g_{22,1}=2g_{12,2}$\ on $%
G_{T}^{1}$. Then $\Gamma ^{1}=0$\ is equivalent to (see\ \cite{10,17})

\begin{equation}
g_{12,2}=\frac{1}{2}g_{12}\frac{\Omega _{,2}}{\Omega }.  \tag{3.6}
\label{5.40}
\end{equation}%
The equation

\begin{equation}
\frac{1}{4}g_{,2}^{\alpha \beta }g_{\alpha \beta ,2}-\frac{1}{2}\left(
g^{\alpha \beta }g_{\alpha \beta ,2}\right) _{,2}=T_{22},  \tag{3.7}
\label{5.41}
\end{equation}%
provides the following non linear second order ODE with the conformity
factor $\Omega $\ as unknown

\begin{equation}
-\left( \frac{\Omega _{,2}}{\Omega }\right) ^{2}+\frac{1}{2}h_{\alpha \beta
,2}h_{,2}^{\alpha \beta }-2\left( \frac{\Omega _{,2}}{\Omega }\right)
_{,2}=2K_{22}\left( g_{12}\right) ^{4}.  \tag{3.8}  \label{5.42}
\end{equation}%
If we set $\Omega =e^{V}$,\ then the following system of ODE is derived from
$(\ref{5.40})$ and $(\ref{5.42})$ in order to determine the conformity
factor together with $g_{12}$

\begin{equation}
\begin{array}{l}
2V_{,22}=-\left( V_{,2}\right) ^{2}-2K_{22}\left( g_{12}\right) ^{4}+\frac{1%
}{2}h_{\alpha \beta ,2}h_{,2}^{\alpha \beta }, \\
g_{12,2}=\frac{1}{2}g_{12}V_{,2}.%
\end{array}
\tag{3.9}  \label{5.43}
\end{equation}%
Let $T\in \left( 0,\infty \right) $. Assume $h_{\alpha \beta },$ $K_{22}\in
C^{\infty }\left( G_{T}^{1}\right) .$ Take $V_{0}$, $V_{1},$ $W_{0}\in
C^{\infty }\left( \Gamma \right) ,$ where $\Gamma \equiv G_{T}^{1}\cap
G_{T}^{2}$. Then there exists $T_{1}\in \left] 0,T\right] $ such that $%
\left( \ref{5.43}\right) $\ has a unique solution $\left( V,g_{12}\right)
\in C^{\infty }\left( G_{T_{1}}^{1}\right) \times C^{\infty }\left(
G_{T_{1}}^{1}\right) $\ satisfying $V=V_{0}$, $V_{,2}=V_{1}$, $g_{12}=W_{0}$%
\ on $\Gamma $. This follows from known local existence and uniqueness
results concerning non-linear ODE with $C^{\infty }$ data in Banach spaces.

\subsubsection{The condition $g_{22,1}-2g_{12,2}=0$ on $G_{T_{1}}^{1}$}

On $G_{T_{1}}^{1}$,\ the reduced equation $\widetilde{R}_{22}=T_{22}$\ is
equivalent to the following homogenous ODE with unknown $g_{22,1}-2g_{12,2}$
(see \cite{10,17})

\begin{equation}
\left( g^{12}\right) ^{2}g_{12,2}\left( g_{22,1}-2g_{12,2}\right)
-g^{12}\left( g_{22,1}-2g_{12,2}\right) _{,2}=0.  \tag{3.10}  \label{5.44}
\end{equation}%
Assume $g_{22,1}-2g_{12,2}=0$\ on $\Gamma $.\ Then $g_{22,1}-2g_{12,2}=0$\
on $G_{T_{1}}^{1}$ and so $\Gamma ^{1}=0$\ on $G_{T_{1}}^{1}$.

\subsection{Relations $\Gamma ^{\protect\alpha }=0$\ on $G_{T_{1}}^{1}$}

We seek for a combination between $R_{2\alpha }$ and $\Gamma ^{\alpha }$
that will provide an homogenous ODE on $G^{1}$ with unknown $\Gamma ^{\alpha
}$.

\begin{proposition}
On $G_{T}^{1}$,\ it holds that%
\begin{equation}
R_{2\alpha }+\frac{1}{2}g_{\alpha \beta }\Gamma _{,2}^{\beta }+\left(
g^{12}g_{12,2}g_{\alpha \beta }+\frac{1}{2}g_{\alpha \beta ,2}\right) \Gamma
^{\beta }=\psi _{\alpha },  \tag{3.11}  \label{5.45}
\end{equation}%
\begin{equation}
T_{2\alpha }=-\frac{\left( -g_{12}\right) ^{\frac{7}{2}}h_{\alpha 3}\Omega ^{%
\frac{1}{2}}\left( h_{33}\right) ^{-\frac{1}{2}}}{8\sqrt{2}m^{2}}K_{3}-\frac{%
\left( -g_{12}\right) ^{\frac{7}{2}}\left( h_{\alpha 4}\Omega ^{\frac{1}{2}%
}\left( h_{33}\right) ^{\frac{1}{2}}-\Omega h_{\alpha 3}h_{34}\right) }{8%
\sqrt{2}m^{2}}K_{4},  \tag{3.12}  \label{5.46}
\end{equation}%
where%
\begin{equation}
\begin{array}{l}
\psi _{\alpha }=\frac{1}{2}\left( g^{12}\right) g_{12,2}\left[
-2g^{12}g_{12,\alpha }+g^{\mu \theta }\left( 2g_{\alpha \mu ,\theta }-g_{\mu
\theta ,\alpha }\right) \right] \\
\text{ \ \ \ \ \ }+\frac{1}{4}g_{\alpha \beta ,2}\ \left[ -2g^{\beta \lambda
}g^{12}g_{12,\lambda }+g^{\beta \lambda }g^{\mu \theta }\left( 2g_{\lambda
\mu ,\theta }-g_{\mu \theta ,\lambda }\right) \right] \\
\text{ \ \ \ \ \ }+\frac{1}{2}\left( g^{\lambda \beta }g_{\alpha \beta
,2}\right) _{,\lambda }-3\left( g^{12}g_{12,2}\right) _{,\alpha }-\frac{3}{2}%
\left( g^{12}\right) ^{2}g_{12,2}g_{12,\alpha } \\
\text{ \ \ \ \ \ }+\frac{1}{2}g^{12}\left( g_{22,1\alpha }+g_{12,2\alpha
}\right) \\
\text{ \ \ \ \ \ }+\frac{1}{2}g_{,2}^{\beta \lambda }\left( g_{\lambda \beta
,\alpha }+g_{\lambda \alpha ,\beta }\right) +\left( g^{12}\right)
^{2}g_{12,2}g_{12,\alpha } \\
\text{ \ \ \ \ \ }+\frac{1}{2}g_{\alpha \beta }\left[ -2g^{\beta \lambda
}g^{12}g_{12,\lambda }+g^{\beta \lambda }g^{\mu \theta }\left( 2g_{\lambda
\mu ,\theta }-g_{\mu \theta ,\lambda }\right) \right] _{,2}, \\
K_{3}=\int_{B}\varphi \left( x,v\right) \left( 1+v^{2}\right) v^{3}d^{3}v,%
\text{\quad }K_{4}=\int_{B}\varphi \left( x,v\right) \left( 1+v^{2}\right)
v^{4}d^{3}v.%
\end{array}
\tag{3.13}  \label{5.47}
\end{equation}
\end{proposition}

\begin{proof}
See appendix B.
\end{proof}

The relations $\Gamma ^{3}=\Gamma ^{4}=0$\ on $G^{1}$ is to be arranged
under a suitable choice of the distribution function on $\widehat{G^{1}}%
=G^{1}\times B$.\ It is at this level that the Rendall method need to be
modified. Assume that the distribution function $\varphi $ is such that%
\begin{equation}
T_{2\alpha }=\psi _{\alpha }\text{ on }G_{T_{1}}^{1}.  \tag{3.14}
\label{5.48}
\end{equation}%
Then the reduced system $\widetilde{R}_{2\alpha }=T_{2\alpha }$\ is
equivalent to the following homogenous system of ODE on $G_{T_{1}}^{1}$\
with unknown $\left( \Gamma ^{3},\Gamma ^{4}\right) $

\begin{equation}
\begin{array}{l}
g_{3\beta }\Gamma _{,2}^{\beta }+\left( g^{12}g_{12,2}g_{3\beta }+\frac{1}{2}%
g_{3\beta ,2}\right) \Gamma ^{\beta }=0, \\
g_{4\beta }\Gamma _{,2}^{\beta }+\left( g^{12}g_{12,2}g_{4\beta }+\frac{1}{2}%
g_{4\beta ,2}\right) \Gamma ^{\beta }=0.%
\end{array}
\tag{3.15}  \label{5.49}
\end{equation}%
Assumption $\left( \ref{5.48}\right) $ is an integral system on $%
G_{T_{1}}^{1}$ in the sense that it is written explicitly as follows

\begin{equation}
A_{\alpha }U+B_{\alpha }V=\psi _{\alpha },  \tag{3.16}  \label{5.50}
\end{equation}%
where

\begin{equation}
\begin{array}{l}
A_{\alpha }=-\frac{\left( -g_{12}\right) ^{\frac{7}{2}}h_{\alpha 3}\Omega ^{%
\frac{1}{2}}\left( h_{33}\right) ^{-\frac{1}{2}}}{8\sqrt{2}m^{2}},\quad
B_{\alpha }=-\frac{\left( -g_{12}\right) ^{\frac{7}{2}}\Omega ^{\frac{1}{2}%
}\left( h_{\alpha 4}\left( h_{33}\right) ^{\frac{1}{2}}-\Omega ^{\frac{1}{2}%
}h_{\alpha 3}h_{34}\right) }{8\sqrt{2}m^{2}} \\
U=\int_{B}\varphi \left( x,v\right) \left( 1+v^{2}\right) v^{3}d^{3}v,\quad
V=\int_{B}\varphi \left( x,v\right) \left( 1+v^{2}\right) v^{4}d^{3}v.%
\end{array}
\tag{3.17}  \label{5.51}
\end{equation}%
The determinant of the system $\left( \ref{5.50}\right) $ is equal to $\frac{%
\left( -g_{12}\right) ^{7}\Omega }{128m^{4}}$ on $G_{T_{1}}^{1}$. So the
solutions are given by the relations below

\begin{equation}
\begin{array}{c}
U=\frac{128m^{4}}{\left( g_{12}\right) ^{7}\Omega }\left( \psi
_{4}B_{3}-\psi _{3}B_{4}\right) , \\
V=\frac{128m^{4}}{\left( g_{12}\right) ^{7}\Omega }\left( \psi
_{3}A_{4}-\psi _{4}A_{3}\right) .%
\end{array}
\tag{3.18}  \label{5.52}
\end{equation}%
In brief $\varphi $ has to be chosen in such a way that the following
integral system $\left( IS\right) $ holds for every $x\in $ $G_{T_{1}}^{1}$

\begin{equation}
\begin{array}{c}
\int_{B}\varphi \left( x,v\right) \left( 1+v^{2}\right) v^{3}d^{3}v=\frac{%
128m^{4}}{\left( g_{12}\right) ^{7}\Omega }\left( \psi _{4}B_{3}-\psi
_{3}B_{4}\right) , \\
\int_{B}\varphi \left( x,v\right) \left( 1+v^{2}\right) v^{4}d^{3}v=\frac{%
128m^{4}}{\left( g_{12}\right) ^{7}\Omega }\left( \psi _{3}A_{4}-\psi
_{4}A_{3}\right) .%
\end{array}
\tag{IS}
\end{equation}%
\

Assume $\Gamma ^{\beta }=0$\ on $\Gamma $.\ Then, in view of $\left( \ref%
{5.49}\right) $, $\Gamma ^{\beta }=0$\ on $G_{T_{1}}^{1}$.

\subsection{Relation $\Gamma ^{2}=0$\ on $G_{T_{1}}^{1}$}

We seek for a combination of $g^{\alpha \beta }R_{\alpha \beta }$, $\Gamma
^{2}$ and $\Gamma _{,2}^{2}$\ that will provide an homogenous ODE on $%
G_{T_{1}}^{1}$\ with unknown $\Gamma ^{2}$.

\subsubsection{Combination of $g^{\protect\alpha \protect\beta }R_{\protect%
\alpha \protect\beta }$, $\Gamma ^{2}$\ and $\Gamma _{,2}^{2}$}

\begin{proposition}
On $G_{T_{1}}^{1}$\ the following combination holds%
\begin{equation}
g^{\alpha \beta }R_{\alpha \beta }-2\Gamma _{,2}^{2}-2g^{12}g_{12,2}\Gamma
^{2}=\frac{1}{4}g^{\alpha \beta }\left( N_{\alpha \beta }+M_{\alpha \beta
}\right) ,  \tag{3.21}  \label{5.55}
\end{equation}%
where%
\begin{equation}
\begin{array}{l}
N_{\alpha \beta }=-g^{12}g^{\lambda \mu }g_{2\lambda ,1}\left( g_{\beta \mu
,\alpha }+g_{\mu \alpha ,\beta }-g_{\alpha \beta ,\mu }\right) \\
+g^{12}\left( g_{2\beta ,1\alpha }+g_{2\alpha ,1\beta }\right) -\left[
2g^{12}g_{12,\alpha }+g^{\lambda \mu }\left( g_{\mu \lambda ,\alpha }+g_{\mu
\alpha ,\lambda }-g_{\alpha \lambda ,\mu }\right) \right] _{,\beta },%
\end{array}
\tag{3.22}  \label{5.56}
\end{equation}%
\begin{equation}
\begin{array}{l}
M_{\alpha \beta }=\left[ 2g^{12}g_{12,\lambda }+g^{\mu \theta }\left( g_{\mu
\theta ,\lambda }+g_{\theta \lambda ,\mu }-g_{\mu \lambda ,\theta }\right) %
\right] \left[ g^{\mu \lambda }\left( g_{\mu \beta ,\alpha }+g_{\mu \alpha
,\beta }-g_{\alpha \beta ,\mu }\right) \right] \\
-\left( g^{12}\right) ^{2}\left( g_{12,\beta }+g_{2\beta ,1}\right) \left(
g_{12,\alpha }+g_{2\alpha ,1}\right) \\
-\left[ g^{12}\left( g_{12,\beta }-g_{2\beta ,1}\right) \right] \left[
g^{12}\left( g_{12,\alpha }-g_{2\alpha ,1}\right) \right] \\
-\left[ g^{\theta \mu }\left( g_{\theta \lambda ,\beta }+g_{\theta \beta
,\lambda }-g_{\lambda \beta ,\theta }\right) \right] \left[ g^{\delta
\lambda }\left( g_{\delta \mu ,\alpha }+g_{\delta \alpha ,\mu }-g_{\alpha
\mu ,\delta }\right) \right] .%
\end{array}
\tag{3.23}  \label{5.57}
\end{equation}
\end{proposition}

\begin{proof}
It follows from relations $\left( 7.58-61\right) $ of \cite{10} by using $%
\left( \ref{5.36b}\right) $ and $\left( \ref{5.37}\right) $.
\end{proof}

\subsubsection{Expression of $g^{\protect\alpha \protect\beta }T_{\protect%
\alpha \protect\beta }$}

\begin{proposition}
On $G_{T_{1}}^{1}$\ it holds that%
\begin{equation}
\begin{array}{l}
g^{\alpha \beta }T_{\alpha \beta }=\frac{\left( -g_{12}\right) ^{3}}{8m^{2}}%
\int_{B}\varphi \left( x,v\right) \left( v^{3}\right) ^{2}d^{3}v \\
\text{ \ \ \ \ \ \ \ \ \ \ \ \ }+\frac{\left( -g_{12}\right) ^{3}h_{34}}{%
4m^{2}}\left( 1-\left( \Omega h_{33}\right) ^{\frac{1}{2}}\right)
\int_{B}\varphi \left( x,v\right) v^{3}v^{4}d^{3}v \\
\text{ \ \ \ \ \ \ \ \ \ \ \ \ }+\frac{\left( -g_{12}\right) ^{3}}{8m^{2}}%
\left[ \left( h_{34}\right) ^{2}\left( \Omega h_{33}\right) ^{\frac{1}{2}%
}\left( \left( \Omega h_{33}\right) ^{\frac{1}{2}}-2\right) +h_{44}h_{33}%
\right] \int_{B}\varphi \left( x,v\right) \left( v^{4}\right) ^{2}d^{3}v.%
\end{array}
\tag{3.24}  \label{5.58}
\end{equation}
\end{proposition}

\begin{proof}
See appendix C.
\end{proof}

In addition to assumption $\left( \ref{5.48}\right) $, assume%
\begin{equation}
g^{\alpha \beta }T_{\alpha \beta }=\frac{1}{4}g^{\alpha \beta }\left(
N_{\alpha \beta }+M_{\alpha \beta }\right) \text{ on }G_{T_{1}}^{1}.
\tag{3.25}  \label{5.59}
\end{equation}%
Then, in\ view\ of\ $\left( \ref{5.55}\right) $, the reduced system $%
\widetilde{R}_{\alpha \beta }=T_{\alpha \beta }$ provides the following
homogenous ODE on $G_{T_{1}}^{1}$\ with unknown $\Gamma ^{2}$%
\begin{equation}
-2\Gamma _{,2}^{2}-2g^{12}g_{12,2}\Gamma ^{2}=0.  \tag{3.26}  \label{5.60}
\end{equation}%
The assumption $\left( \ref{5.59}\right) $\ is a supplementary integral
equation\ which is written explicitly as follows, $\forall x\in $\ $%
G_{T_{1}}^{1}$

\begin{equation}
\begin{array}{l}
\frac{\left( -g_{12}\right) ^{3}}{8m^{2}}\int_{B}\varphi \left( x,v\right)
\left( v^{3}\right) ^{2}d^{3}v+\frac{\left( -g_{12}\right) ^{3}h_{34}}{4m^{2}%
}\left( 1-\left( \Omega h_{33}\right) ^{\frac{1}{2}}\right) \int_{B}\varphi
\left( x,v\right) v^{3}v^{4}d^{3}v \\
+\frac{\left( -g_{12}\right) ^{3}}{8m^{2}}\left[ \left( h_{34}\right)
^{2}\left( \Omega h_{33}\right) ^{\frac{1}{2}}\left( \left( \Omega
h_{33}\right) ^{\frac{1}{2}}-2\right) +h_{44}h_{33}\right] \int_{B}\varphi
\left( x,v\right) \left( v^{4}\right) ^{2}d^{3}v \\
=\frac{1}{4}g^{\alpha \beta }\left( N_{\alpha \beta }+M_{\alpha \beta
}\right) .%
\end{array}
\tag{IE}
\end{equation}%
In view of $\left( \ref{5.60}\right) $, assuming $\Gamma ^{2}=0$\ on $\Gamma
$ gives $\Gamma ^{2}=0$\ on $G_{T_{1}}^{1}$.

What we\ have just proved for the resolution of the constraints problem
associated to the EV system can be summed up in the following main theorem.

\begin{theorem}
Under the suitable integral assumptions $\left( IS\right) $ and $\left(
IE\right) $\ on the distribution function, there exists initial data for the
reduced EV system such that the constraints $\Gamma ^{k}=0$\ are satisfied
on $G^{1}\cup G^{2}$ for the corresponding solution of the evolution problem
associated to the EV system.
\end{theorem}

\begin{remark}
It would be interesting to see whether the constraints integral equations
can be avoided. One way of doing this is to work in temporal gauge and null
moving frame (see \cite{15}).\ But in the harmonic gauge case the issue may
not be evident. Nevertheless we think that one could use an orthonormal
frame in order to avoid that the metric appears (in an involved way) in the
energy-momentum tensor (see \cite{16}).
\end{remark}

\section*{Appendix A: Proof of Proposition 6}

The first equality of $\left( \ref{5.38}\right) $\ is provided in relation $%
\left( 7.42\right) $ of \cite{10}. We handle the second one by using the
expression $\left( \ref{5.30}\right) $ of the energy-momentum tensor given
in Proposition 5 to have
\begin{equation}
T_{ij}\left( x\right) =\frac{1}{2m^{2}}\int_{B}\varphi \left( x,v\right)
v_{ij}\left\vert g\right\vert ^{\frac{1}{2}}\left( -\widehat{g}^{11}\right)
^{-\frac{3}{2}}\left\vert \widetilde{g}\right\vert ^{-\frac{1}{2}}d^{3}v,
\tag{A.1}
\end{equation}%
where $v_{ij}=\frac{p_{i}p_{j}}{\left( q^{1}\right) ^{2}}$. For $i=j=2$ $%
\left( A.1\right) $\ reads
\begin{equation}
T_{22}\left( x\right) =\frac{1}{2m^{2}}\int_{B}\varphi \left( x,v\right)
v_{22}\left\vert g\right\vert ^{\frac{1}{2}}\left( -\widehat{g}^{11}\right)
^{-\frac{3}{2}}\left\vert \widetilde{g}\right\vert ^{-\frac{1}{2}}d^{3}v.
\tag{A.2}
\end{equation}%
From $\left( \ref{5.36a}\right) $ we gain%
\begin{equation}
p_{2}=\frac{1}{2}g_{12}q^{1}\left( 1+w^{2}\right) \text{ on }G^{1}.
\tag{A.3}
\end{equation}%
Thus
\begin{equation}
v_{22}=\frac{1}{4}\left( g_{12}\right) ^{2}\left( 1+w^{2}\right) ^{2}\text{
on }G^{1}.  \tag{A.4}
\end{equation}%
$w^{2}$ is now expressed on $G^{1}$\ in terms of $v^{A}$ via $\left( \ref%
{5.21}\right) $ to give%
\begin{equation}
w^{2}=\left( \lambda _{B}^{2}\right) ^{-1}v^{B}\left( -\widehat{g}%
^{11}\right) ^{-\frac{1}{2}}-\widetilde{g}^{12}.  \tag{A.5}
\end{equation}%
The exact expression of $\left( \lambda _{B}^{A}\right) $ is needed.
Splitting the quadratic form $\left( \ref{5.18}\right) $ yields
\begin{equation}
\begin{array}{l}
\widehat{g}_{AB}X^{A}X^{B} \\
=\left[ \left( \widehat{g}_{22}\right) ^{\frac{1}{2}}X^{2}+\left( \widehat{g}%
_{22}\right) ^{-\frac{1}{2}}\widehat{g}_{2\alpha }X^{\alpha }\right] ^{2} \\
+\left[ \left( \widehat{g}_{33}-\left( \widehat{g}_{22}\right) ^{-1}\left(
\widehat{g}_{23}\right) ^{2}\right) ^{\frac{1}{2}}X^{3}+\left( \widehat{g}%
_{22}\widehat{g}_{34}-\widehat{g}_{23}\widehat{g}_{24}\right) \left(
\widehat{g}_{22}\left( \widehat{g}_{22}\widehat{g}_{33}-\left( \widehat{g}%
_{23}\right) ^{2}\right) \right) ^{-\frac{1}{2}}X^{4}\right] ^{2} \\
+\left[ \left( \left( \widehat{g}_{22}\widehat{g}_{33}-\left( \widehat{g}%
_{23}\right) ^{2}\right) \left( \widehat{g}_{22}\widehat{g}_{44}-\left(
\widehat{g}_{24}\right) ^{2}\right) -\left( \widehat{g}_{22}\widehat{g}_{34}-%
\widehat{g}_{23}\widehat{g}_{24}\right) ^{2}\right) ^{\frac{1}{2}}\left(
\widehat{g}_{22}\left( \widehat{g}_{22}\widehat{g}_{33}-\left( \widehat{g}%
_{23}\right) ^{2}\right) \right) ^{-\frac{1}{2}}X^{4}\right] ^{2}.%
\end{array}
\tag{A.6}
\end{equation}%
By computing $\widehat{g}_{ij}$ via tensorial transformation formulae we
gain
\begin{equation*}
\left( \widehat{g}_{ij}\right) =\left(
\begin{array}{cccc}
\frac{1}{4}g_{11}+\frac{1}{2}g_{12}+\frac{1}{4}g_{22} & \frac{1}{4}g_{11}-%
\frac{1}{4}g_{22} & \frac{1}{2}\left( g_{13}+g_{23}\right) & \frac{1}{2}%
\left( g_{14}+g_{24}\right) \\
\frac{1}{4}g_{11}-\frac{1}{4}g_{22} & \frac{1}{4}g_{11}-\frac{1}{2}g_{12}+%
\frac{1}{4}g_{22} & \frac{1}{2}\left( g_{13}-g_{23}\right) & \frac{1}{2}%
\left( g_{14}-g_{24}\right) \\
\frac{1}{2}\left( g_{13}+g_{23}\right) & \frac{1}{2}\left(
g_{13}-g_{23}\right) & g_{33} & g_{34} \\
\frac{1}{2}\left( g_{14}+g_{24}\right) & \frac{1}{2}\left(
g_{14}-g_{24}\right) & g_{34} & g_{44}%
\end{array}%
\right) .
\end{equation*}%
From $\left( \ref{5.36a}\right) $ and\ $\left( \ref{5.36b}\right) $\ we get
\begin{equation}
\left( \widehat{g}_{ij}\right) =\left(
\begin{array}{cccc}
\frac{1}{2}g_{12} & 0 & 0 & 0 \\
0 & -\frac{1}{2}g_{12} & 0 & 0 \\
0 & 0 & g_{33} & g_{34} \\
0 & 0 & g_{34} & g_{44}%
\end{array}%
\right) \text{ on }G^{1}.  \tag{A.7}
\end{equation}%
Hence,\ in\ view\ of\ $\left( \ref{5.18}\right) $,\ we\ gain%
\begin{equation}
\left( \lambda _{B}^{A}\right) =\left(
\begin{array}{ccc}
\left( -\frac{1}{2}g_{12}\right) ^{\frac{1}{2}} & 0 & 0 \\
0 & \left( \Omega h_{33}\right) ^{\frac{1}{2}} & \Omega h_{34} \\
0 & 0 & \Omega ^{\frac{1}{2}}h_{33}^{-\frac{1}{2}}%
\end{array}%
\right) \text{ on }G^{1}.  \tag{A.8}
\end{equation}%
Thus, the inverse matrix of $\left( \lambda _{B}^{A}\right) $ reads%
\begin{equation}
\left( \lambda _{B}^{A}\right) ^{-1}=\left(
\begin{array}{ccc}
\left( -\frac{1}{2}g_{12}\right) ^{-\frac{1}{2}} & 0 & 0 \\
0 & \left( \Omega h_{33}\right) ^{-\frac{1}{2}} & -h_{34} \\
0 & 0 & \Omega ^{-\frac{1}{2}}h_{33}^{\frac{1}{2}}%
\end{array}%
\right) \text{ on }G^{1}.  \tag{A.9}
\end{equation}%
Simple calculation gives
\begin{equation}
\widehat{g}^{11}=2\left( g_{12}\right) ^{-1}\text{ on }G^{1}.  \tag{A.10}
\end{equation}%
It is worth noting from\ $\left( A.10\right) $\ that the condition $\widehat{%
g}^{11}<0$ is equivalent to $g_{12}<0$ on $G^{1}$. We also have
\begin{equation}
\widetilde{g}^{12}=0\text{ on }G^{1}.  \tag{A.11}
\end{equation}%
$\left( A.5\right) $, $\left( A.9\right) $, $\left( A.10\right) $\ and $%
\left( A.11\right) $\ imply%
\begin{equation}
w^{2}=v^{2}\text{ on }G^{1}.  \tag{A.12}
\end{equation}%
It follows from $\left( A.4\right) $ and\ $\left( A.12\right) $\ that
\begin{equation}
v_{22}=\frac{1}{4}\left( g_{12}\right) ^{2}\left( 1+v^{2}\right) ^{2}\text{
on }G^{1}.  \tag{A.13}
\end{equation}%
We now handle the term $\left\vert g\right\vert ^{\frac{1}{2}}\left( -%
\widehat{g}^{11}\right) ^{-\frac{3}{2}}\left\vert \widetilde{g}\right\vert
^{-\frac{1}{2}}$ on $G^{1}$. From $\left( \ref{5.36a}\right) $ and $\left( %
\ref{5.36b}\right) $ we gain
\begin{equation}
\det \left( g_{ij}\right) =-\left( \Omega g_{12}\right) ^{2}\text{ on }G^{1}.
\tag{A.14}
\end{equation}%
Thus
\begin{equation}
\left\vert g\right\vert ^{\frac{1}{2}}=\Omega \left\vert g_{12}\right\vert
=-\Omega g_{12}\text{ on }G^{1}.  \tag{A.15}
\end{equation}%
Using $\left( \ref{5.36a}\right) $, $\left( \ref{5.36b}\right) $ and
tensorial transformation formulae we gain
\begin{equation*}
\widehat{g}_{22}=-\frac{1}{2}g_{12},\quad \widehat{g}_{2\alpha }=0,\text{%
\quad }\widehat{g}_{\alpha \beta }=\Omega h_{\alpha \beta }\text{ on }G^{1}.
\end{equation*}%
Hence
\begin{equation}
\widetilde{g}=\det \left( \widehat{g}_{AB}\right) =-\frac{1}{2}g_{12}\Omega
^{2}\text{ on }G^{1}.  \tag{A.16}
\end{equation}%
From $\left( A.10\right) $, $\left( A.15\right) $\ and $\left( A.16\right) $%
\ we gain%
\begin{equation}
\left\vert g\right\vert ^{\frac{1}{2}}\left( -\widehat{g}^{11}\right) ^{-%
\frac{3}{2}}\left\vert \widetilde{g}\right\vert ^{-\frac{1}{2}}=\frac{1}{2}%
\left( g_{12}\right) ^{2}\text{ on }G^{1}.  \tag{A.17}
\end{equation}%
The insertion of\ $\left( A.13\right) $\ and $\left( A.17\right) $ into $%
\left( A.2\right) \ $gives the desired expression of $T_{22}\left( x\right) $%
.

\section*{Appendix B: Proof of Proposition 7}

$\left( \ref{5.45}\right) $\ follows from relation $\left( 7.50\right) $ of
\cite{10} by using $g_{1\lambda }=0$ on $G^{1}$. $T_{2\alpha }\left(
x\right) $ is\ computed\ in\ the\ same\ manner as\ $T_{22}\left( x\right) $.
Actually, for $i=2$ and $j=\alpha $, $\left( A.1\right) $ reads
\begin{equation}
T_{2\alpha }\left( x\right) =\frac{1}{2m^{2}}\int_{B}\varphi \left(
x,v\right) v_{2\alpha }\left\vert g\right\vert ^{\frac{1}{2}}\left( -%
\widehat{g}^{11}\right) ^{-\frac{3}{2}}\left\vert \widetilde{g}\right\vert
^{-\frac{1}{2}}d^{3}v.  \tag{B.1}
\end{equation}%
From $\left( \ref{5.36a}\right) $ and $\left( \ref{5.36b}\right) $\ we gain%
\begin{equation}
p_{\alpha }=q^{1}g_{\alpha \beta }w^{\beta }\text{ on }G^{1}.  \tag{B.2}
\end{equation}%
$\left( A.3\right) $ and $\left( B.2\right) $\ yields
\begin{equation}
v_{2\alpha }=\frac{1}{2}g_{12}\left( 1+v^{2}\right) \left( g_{\alpha
3}w^{3}+g_{\alpha 4}w^{4}\right) \text{ on }G^{1}.  \tag{B.3}
\end{equation}%
From $\left( \ref{5.21}\right) $ and $\left( A.9\right) $ we gain
\begin{equation}
w^{3}=\frac{1}{\sqrt{2}}\left( -g_{12}\right) ^{\frac{1}{2}}\left[ \left(
\Omega h_{33}\right) ^{-\frac{1}{2}}v^{3}-h_{34}v^{4}\right] ,\text{\quad }%
w^{4}=\frac{1}{\sqrt{2}}\left( -g_{12}\right) ^{\frac{1}{2}}\Omega ^{-\frac{1%
}{2}}\left( h_{33}\right) ^{\frac{1}{2}}v^{4}.  \tag{B.4}
\end{equation}%
$\left( B.3\right) $ and $\left( B.4\right) $ imply%
\begin{equation}
v_{2\alpha }=-\frac{1}{2\sqrt{2}}\left( 1+v^{2}\right) \left( -g_{12}\right)
^{\frac{3}{2}}\Omega \left[ h_{\alpha 3}\left( \left( \Omega h_{33}\right)
^{-\frac{1}{2}}v^{3}-h_{34}v^{4}\right) +h_{\alpha 4}\Omega ^{-\frac{1}{2}%
}\left( h_{33}\right) ^{\frac{1}{2}}v^{4}\right] \text{ on }G^{1}.  \tag{B.5}
\end{equation}%
Insertion$\ $of $\left( A.17\right) $ and $\left( B.5\right) $ into $\left(
B.1\right) \ $gives the expression $\left( \ref{5.46}\right) $ of $%
T_{2\alpha }\left( x\right) $.

\section*{Appendix C: Proof of Proposition 9}

$T_{\alpha \beta }\left( x\right) $ is handled in the same way as $T_{22}$
and $T_{2\alpha }$. From $\left( A.1\right) $ we have
\begin{equation}
T_{\alpha \beta }\left( x\right) =\frac{1}{2m^{2}}\int_{B}\varphi \left(
x,v\right) v_{\alpha \beta }\left\vert g\right\vert ^{\frac{1}{2}}\left( -%
\widehat{g}^{11}\right) ^{-\frac{3}{2}}\left\vert \widetilde{g}\right\vert
^{-\frac{1}{2}}d^{3}v.  \tag{C.1}
\end{equation}%
$\left( B.2\right) $ implies
\begin{equation}
v_{\alpha \beta }=g_{\alpha \lambda }g_{\beta \mu }w^{\lambda }w^{\mu }\text{
on }G^{1}.  \tag{C.2}
\end{equation}%
Insertion of $\left( A.17\right) $\ and $\left( C.2\right) $\ into $\left(
C.1\right) $\ gives
\begin{equation*}
T_{\alpha \beta }\left( x\right) =\frac{\left( g_{12}\right) ^{2}g_{\alpha
\lambda }g_{\beta \mu }}{4m^{2}}\int_{B}\varphi \left( x,v\right) w^{\lambda
}w^{\mu }d^{3}v.
\end{equation*}%
Thus%
\begin{equation}
\begin{array}{l}
g^{\alpha \beta }T_{\alpha \beta }\left( x\right) =\frac{\left(
g_{12}\right) ^{2}\Omega h_{33}}{4m^{2}}\int_{B}\varphi \left( x,v\right)
\left( w^{3}\right) ^{2}d^{3}v+\frac{\left( g_{12}\right) ^{2}\Omega h_{34}}{%
2m^{2}}\int_{B}\varphi \left( x,v\right) w^{3}w^{4}d^{3}v \\
\text{ \ \ \ \ \ \ \ \ \ \ \ \ \ \ \ \ \ \ }+\frac{\left( g_{12}\right)
^{2}\Omega h_{44}}{4m^{2}}\int_{B}\varphi \left( x,v\right) \left(
w^{4}\right) ^{2}d^{3}v.%
\end{array}
\tag{C.3}
\end{equation}%
The insertion of $\left( B.4\right) $\ into $\left( C.3\right) $\ gives the
expression $\left( \ref{5.58}\right) $ of $g^{\alpha \beta }T_{\alpha \beta
}\left( x\right) $ on $G^{1}$.

\section*{Appendix D: Discussion on the integral constraints equations}

\subsection*{The integral system $\left( IS\right) $}

The integral system $\left( IS\right) $ is written as follows%
\begin{equation}
\langle \varphi _{x},f\rangle =h_{1}\left( x\right) ,\quad \langle \varphi
_{x},g\rangle =h_{2}\left( x\right) .  \tag{D.1}
\end{equation}%
where $\langle ,\rangle $ denotes the scalar product in $L^{2}\left(
B\right) $ and%
\begin{equation}
\begin{array}{l}
f\left( v\right) =\left( 1+v^{2}\right) v^{3},\text{\quad }g\left( v\right)
=\left( 1+v^{2}\right) v^{4}, \\
h_{1}=\frac{128m^{4}}{\left( -g_{12}\right) ^{7}\Omega }\left( \psi
_{3}B_{4}-\psi _{4}B_{3}\right) ,\text{\quad }h_{2}=\frac{128m^{4}}{\left(
-g_{12}\right) ^{7}\Omega }\left( \psi _{4}A_{3}-\psi _{3}A_{4}\right) .%
\end{array}
\tag{D.2}
\end{equation}%
As the distribution function must be non-negative, let us seek $\varphi $ of
the form
\begin{equation}
\varphi _{x}\left( v\right) \equiv \varphi \left( x,v\right) =\left[ a\left(
x\right) f\left( v\right) +b\left( x\right) g\left( v\right) +c\left(
x\right) \right] ^{2},  \tag{D.3}
\end{equation}%
where $a\left( x\right) ,$ $b\left( x\right) $ and $c\left( x\right) $ are
unknown functions defined on $G^{1}$. Expanding $\left( D.3\right) $, we
have
\begin{equation}
\varphi =a^{2}f^{2}+b^{2}g^{2}+2abfg+2acf+2bcg+c^{2},  \tag{D.4}
\end{equation}%
where variables have been dropped for simplicity. Let $\varphi $ be like in $%
\left( D.4\right) $. It holds that%
\begin{equation}
\begin{array}{l}
\langle \varphi _{x},f\rangle =a^{2}\langle f^{2},f\rangle +b^{2}\langle
g^{2},f\rangle +2ab\langle fg,f\rangle +2ac\langle f,f\rangle +2bc\langle
g,f\rangle +c^{2}\langle 1,f\rangle , \\
\langle \varphi _{x},g\rangle =a^{2}\langle f^{2},g\rangle +b^{2}\langle
g^{2},g\rangle +2ab\langle fg,g\rangle +2ac\langle f,g\rangle +2bc\langle
g,g\rangle +c^{2}\langle 1,g\rangle .%
\end{array}
\tag{D.5}
\end{equation}%
Spherical coordinates will be used to calculate each of the following
quantities that are needed.%
\begin{equation}
\begin{array}{l}
\langle f^{2},f\rangle =\int_{B}\left[ \left( 1+v^{2}\right) v^{3}\right]
^{3}d^{3}v,\text{\quad }\langle g^{2},f\rangle =\int_{B}\left[ \left(
1+v^{2}\right) v^{4}\right] ^{2}\left( 1+v^{2}\right) v^{3}d^{3}v, \\
\langle fg,f\rangle =\int_{B}\left[ \left( 1+v^{2}\right) v^{3}\right]
^{2}\left( 1+v^{2}\right) v^{4}d^{3}v,\text{\quad }\langle f,f\rangle
=\int_{B}\left[ \left( 1+v^{2}\right) v^{3}\right] ^{2}d^{3}v, \\
\langle g,f\rangle =\int_{B}\left( 1+v^{2}\right) ^{2}v^{3}v^{4}d^{3}v,\text{%
\quad }\langle 1,f\rangle =\int_{B}\left( 1+v^{2}\right) v^{3}d^{3}v,\text{ }
\\
\langle 1,g\rangle =\int_{B}\left( 1+v^{2}\right) v^{4}d^{3}v,\text{\quad }%
\langle g^{2},g\rangle =\int_{B}\left[ \left( 1+v^{2}\right) v^{4}\right]
^{3}d^{3}v,\text{ } \\
\langle g,g\rangle =\int_{B}\left[ \left( 1+v^{2}\right) v^{4}\right]
^{2}d^{3}v.%
\end{array}
\tag{D.6}
\end{equation}%
As $B$ is the open unit ball in $\mathbb{R}^{3}$, we\ set%
\begin{equation}
v^{2}=r\cos \theta \cos \lambda ,\quad v^{3}=r\cos \theta \sin \lambda
,\quad v^{4}=r\sin \theta ,  \tag{D.7}
\end{equation}%
with%
\begin{equation}
0\leq r<1,\quad -\frac{\pi }{2}\leq \theta \leq \frac{\pi }{2},\quad 0\leq
\lambda \leq 2\pi .  \tag{D.8}
\end{equation}%
Define a domain $P$ in $\mathbb{R}^{3}$ as follows%
\begin{equation}
P=\left\{ \left( r,\theta ,\lambda \right) \in \mathbb{R}^{3}/\text{ }0\leq
r<1,\text{ }-\frac{\pi }{2}\leq \theta \leq \frac{\pi }{2},\text{ }0\leq
\lambda \leq 2\pi \right\} .  \tag{D.9}
\end{equation}%
Using the change of variables $\left( D.7\right) $, we gain
\begin{equation}
\begin{array}{l}
\langle f^{2},f\rangle =\int_{P}f^{3}r^{2}\cos \theta dP,\text{\quad }%
\langle g^{2},f\rangle =\int_{P}g^{2}fr^{2}\cos \theta dP, \\
\langle fg,f\rangle =\int_{P}f^{2}gr^{2}\cos \theta dP,\text{\quad }\langle
f,f\rangle =\int_{P}f^{2}r^{2}\cos \theta dP, \\
\langle g,f\rangle =\int_{P}gfr^{2}\cos \theta dP,\text{\quad }\langle
1,f\rangle =\int_{P}fr^{2}\cos \theta dP, \\
\langle 1,g\rangle =\int_{P}gr^{2}\cos \theta dP,\text{\quad }\langle
g^{2},g\rangle =\int_{P}g^{3}r^{2}\cos \theta dP, \\
\langle g,g\rangle =\int_{P}g^{2}r^{2}\cos \theta dP,%
\end{array}
\tag{D.10}
\end{equation}%
where\
\begin{equation*}
dP=drd\theta d\lambda .
\end{equation*}%
After expansion and reduction we integrate the above quantities over $P$ to
obtain%
\begin{equation}
\begin{array}{l}
\langle f^{2},f\rangle =\langle g^{2},f\rangle =\langle f^{2},g\rangle
=\langle g,f\rangle =\langle 1,f\rangle =\langle 1,g\rangle =\langle
g^{2},g\rangle =0, \\
\langle f,f\rangle =\langle g,g\rangle =\frac{32\pi }{105},%
\end{array}
\tag{D.11}
\end{equation}%
From $\left( D.5\right) $ and $\left( D.11\right) $ it holds that%
\begin{equation}
\langle \varphi _{x},f\rangle =\frac{64\pi a\left( x\right) c\left( x\right)
}{105},\quad \langle \varphi _{x},g\rangle =\frac{64\pi b\left( x\right)
c\left( x\right) }{105}.  \tag{D.12}
\end{equation}%
Thus, $\varphi $ solves $\left( D.1\right) $ if and only if%
\begin{equation}
\frac{64\pi a\left( x\right) c\left( x\right) }{105}=h_{1}\left( x\right)
,\quad \frac{64\pi b\left( x\right) c\left( x\right) }{105}=h_{2}\left(
x\right) .  \tag{D.13}
\end{equation}%
In sum $\varphi $ is given by%
\begin{equation}
\begin{array}{l}
\varphi _{x}\left( v\right) \equiv \varphi \left( x,v\right) =\left[ a\left(
x\right) \left( 1+v^{2}\right) v^{3}+b\left( x\right) \left( 1+v^{2}\right)
v^{4}+c\left( x\right) \right] ^{2}, \\
x\in G^{1},\text{ }v=\left( v^{2},v^{3},v^{4}\right) \in B,%
\end{array}
\tag{D.14}
\end{equation}%
where%
\begin{equation}
a\left( x\right) c\left( x\right) =\frac{105h_{1}\left( x\right) }{64\pi }%
,\quad b\left( x\right) c\left( x\right) =\frac{105h_{2}\left( x\right) }{%
64\pi }.  \tag{D.15}
\end{equation}

\subsection*{The integral equation $\left( IE\right) $}

The integral equation $\left( IE\right) $\ is written as follows%
\begin{equation}
E\langle \varphi _{x},e\rangle +I\langle \varphi _{x},i\rangle +S\langle
\varphi _{x},s\rangle =C,  \tag{D.16}
\end{equation}%
where%
\begin{equation}
\begin{array}{l}
E\left( x\right) =\frac{\left( -g_{12}\right) ^{3}}{8m^{2}},\text{\quad }%
e\left( v\right) =\left( v^{3}\right) ^{2}, \\
I\left( x\right) =\frac{\left( -g_{12}\right) ^{3}h_{34}}{4m^{2}}\left(
1-\left( \Omega h_{33}\right) ^{\frac{1}{2}}\right) ,\text{\quad }i\left(
v\right) =v^{3}v^{4}, \\
S\left( x\right) =\frac{\left( -g_{12}\right) ^{3}}{8m^{2}}\left[ \left(
h_{34}\right) ^{2}\left( \Omega h_{33}\right) ^{\frac{1}{2}}\left( \left(
\Omega h_{33}\right) ^{\frac{1}{2}}-2\right) +h_{44}h_{33}\right] , \\
s\left( v\right) =\left( v^{4}\right) ^{2},\text{\quad }C\left( x\right) =%
\frac{1}{4}g^{\alpha \beta }\left( N_{\alpha \beta }+M_{\alpha \beta
}\right) .%
\end{array}
\tag{D.17}
\end{equation}%
Using the expression of $\varphi $ given\ in\ $\left( D.4\right) $ we gain
\begin{equation}
\langle \varphi _{x},e\rangle =a^{2}\langle f^{2},e\rangle +b^{2}\langle
g^{2},e\rangle +2ab\langle fg,e\rangle +2ac\langle f,e\rangle +2bc\langle
g,e\rangle +c^{2}\langle 1,e\rangle .  \tag{D.18}
\end{equation}%
As in the preceding paragraph, the above quantities are found to be%
\begin{equation}
\begin{array}{c}
\langle f^{2},e\rangle =\frac{8\pi }{63},\quad \langle g^{2},e\rangle =\frac{%
8\pi }{189},\quad \langle 1,e\rangle =\frac{4\pi }{15},\quad \langle
fg,e\rangle =\langle f,e\rangle =\langle g,e\rangle =0.%
\end{array}
\tag{D.19}
\end{equation}%
$\left( D.18\right) $ and $\left( D.19\right) $ imply%
\begin{equation}
\langle \varphi _{x},e\rangle =\frac{8\pi }{63}a^{2}+\frac{8\pi }{189}b^{2}+%
\frac{4\pi }{15}c^{2}.  \tag{D.20}
\end{equation}%
Similarly it holds that%
\begin{equation}
\langle \varphi _{x},i\rangle =a^{2}\langle f^{2},i\rangle +b^{2}\langle
g^{2},i\rangle +2ab\langle fg,i\rangle +2ac\langle f,i\rangle +2bc\langle
g,i\rangle +c^{2}\langle 1,i\rangle .  \tag{D.21}
\end{equation}%
Straightforward calculations as above give%
\begin{equation}
\langle f^{2},i\rangle =\langle g^{2},i\rangle =\langle f,i\rangle =\langle
g,i\rangle =\langle 1,i\rangle =0,\quad \langle fg,i\rangle =\frac{8\pi }{189%
}  \tag{D.22}
\end{equation}%
$\left( D.21\right) $ and $\left( D.22\right) $ give%
\begin{equation}
\langle \varphi _{x},i\rangle =\frac{16\pi ab}{189}.  \tag{D.23}
\end{equation}%
We are then left with calculating\
\begin{equation}
\langle \varphi _{x},s\rangle =a^{2}\langle f^{2},s\rangle +b^{2}\langle
g^{2},s\rangle +2ab\langle fg,s\rangle +2ac\langle f,s\rangle +2bc\langle
g,s\rangle +c^{2}\langle 1,s\rangle .  \tag{D.24}
\end{equation}%
By proceeding as above we get%
\begin{equation}
\langle fg,s\rangle =\langle f,s\rangle =\langle g,s\rangle =0,\quad \langle
f^{2},s\rangle =\frac{8\pi }{189},\quad \langle g^{2},s\rangle =\frac{8\pi }{%
63},\quad \langle 1,s\rangle =\frac{4\pi }{15}.  \tag{D.25}
\end{equation}%
$\left( D.24\right) $ and $\left( D.25\right) $ yield
\begin{equation}
\langle \varphi _{x},s\rangle =\frac{8\pi }{189}a^{2}+\frac{8\pi }{63}b^{2}+%
\frac{4\pi }{15}c^{2}.  \tag{D.26}
\end{equation}%
From $\left( D.16\right) ,$ $\left( D.20\right) ,$ $\left( D.23\right) $ and
$\left( D.26\right) $,\ we see that the integral equation $\left( IE\right) $%
\ is equivalent to%
\begin{equation}
10\left( 3E+S\right) a^{2}+10\left( 3S+E\right) b^{2}+20Iab+63\left(
E+S\right) c^{2}=\frac{945}{4\pi }C.  \tag{D.27}
\end{equation}%
In view of $\left( D.15\right) $,\ multiplying $\left( D.27\right) $ by $%
c^{2}$ and\ rearranging\ , we gain
\begin{equation}
\begin{array}{l}
110\,250\left( 3E\left( x\right) +S\left( x\right) \right) \left(
h_{1}\left( x\right) \right) ^{2}+110\,250\left( 3S\left( x\right) +E\left(
x\right) \right) \left( h_{2}\left( x\right) \right) ^{2} \\
+220\,500h_{1}\left( x\right) h_{2}\left( x\right) I\left( x\right)
+258\,048\pi ^{2}\left( E\left( x\right) +S\left( x\right) \right) \left[
c\left( x\right) \right] ^{4} \\
=967\,680\pi C\left( x\right) \left[ c\left( x\right) \right] ^{2}.%
\end{array}
\tag{D.28}
\end{equation}%
$\left( D.28\right) $ is an algebraic equation that can be solved under
suitable assumptions to find $c\left( x\right) $. Doing so we deduce $%
a\left( x\right) $ and $b\left( x\right) $ thanks to $\left( D.15\right) $.
Finally the distribution function $\varphi $ is obtained on $\widehat{G^{1}}$
and has the form $\left( D.14\right) $.\medskip

\textbf{Acknowledgement. }I am thankful to Professor Marcel Dossa who
suggested this topic and guided me throughout the work.

\end{document}